\documentclass[english,11pt,reqno]{amsart}

\usepackage{graphicx}
\usepackage{amsmath}
\usepackage{amssymb}
\usepackage{empheq}
\usepackage{subfig}
\usepackage{natbib}
\usepackage{color}
\usepackage{hyperref}

\newcommand{\tR}{\tilde{\mathbf{R}}}
\newcommand{\tT}{\tilde{\mathbf{T}}}
\newcommand{\tN}{\tilde{\mathbf{N}}}

\newcommand{\tRR}{\tilde{R}}
\newcommand{\tTT}{\tilde{T}}
\newcommand{\tNN}{\tilde{N}}

\newcommand{\bR}{\bar{R}}
\newcommand{\bT}{\bar{T}}
\newcommand{\bN}{\bar{N}}

\newcommand{\tH}{\tilde{\mathbf{H}}}
\newcommand{\tE}{\tilde{\mathbf{E}}}
	
\newcommand{\ta}{\tilde{\mathbf{a}}}
\newcommand{\te}{\tilde{\mathbf{e}}}
\newcommand{\tf}{\tilde{\mathbf{f}}}
\newcommand{\tthet}{\tilde{\mathbf{\theta}}}
\newcommand{\ti}{\tilde{\mathbf{i}}}
\newcommand{\tOmega}{\tilde{\mathbf{\Omega}}}
\newcommand{\tomega}{\tilde{\mathbf{\omega}}}

\theoremstyle{plain}
\newtheorem{theorem}{Theorem}[section]

\newtheorem{lemma}[theorem]{Lemma}

\theoremstyle{definition}

\theoremstyle{remark}

\begin{document}

\title{Stochastic Gauss Equations}
\author{Fr\'ed\'eric Pierret}
\address{SYRTE UMR CNRS 8630, Observatoire de Paris, 77 avenue Denfert-Rochereau, 75014 Paris, France}
\begin{abstract}
We derive the equations of celestial mechanics governing the variations of the orbital elements under a stochastic perturbation generalizing the classical Gauss equations. Explicit formulas are given for the semi-major axis, the eccentricity, the inclination, the longitude of the ascending node, the pericenter angle and the mean anomaly which are express in term of the angular momentum vector $\textbf{H}$ per unit of mass and the energy $E$ per unit of mass. Together, these formulas are called the \emph{stochastic Gauss equations} and they are illustrated numerically on an example from satellite dynamics.
\end{abstract}

\maketitle

\vskip 5mm
\begin{tiny}
	SYRTE UMR CNRS 8630, Observatoire de Paris, France
\end{tiny}

\tableofcontents

\section{Introduction}

Nowadays celestial mechanics is used by a wide class of scientists which provide multiple applications (see \cite{murray_dermott}, \cite{burns1976} and references therein). In all these works, the underling nature of the model considered is always deterministic. However, considering models with randomness or stochastic behavior is not an easy problem (examples for celestial mechanics can be found in \cite{cr}, \cite{cpp} and \cite{bcp}). Indeed, the nature and the origin of such a model needs a real discussion of the phenomena that we want to study.\\

Most of the problems in celestial mechanics are seen as a two-body problem perturbed by a force. For example, the main approach of the $n$-body problem is to consider two bodies, in mutual gravitational interaction, which are perturbed by the other bodies. In that case, the perturbed force is the gravitational attraction of the other bodies.\\

When we are dealing with more than two bodies, or more generally with an arbitrary perturbing force, the orbital elements, which characterize the trajectory of the bodies, do not remain constant. In that case, the main tool of celestial mechanics to study the perturbed problem, is the set of equations given the variations of the orbital elements called, the \textit{Gauss equations}. Because the Gauss equations allow studying general problems in celestial mechanics, we propose in this paper to generalize them to the stochastic case which include by definition the deterministic case. \\

We follow the strategy of \cite{burns1976} who derived Gauss's equations for the elliptical case with elementary considerations which defined the orbital elements in mean of the angular momentum per unit of mass and the energy per unit of mass. From an example of the satellite dynamics, we illustrate numerically the variation of the orbital elements associated. Finally, we give the variation of the Laplace-Runge-Lenz vector. It allows deriving the variation of the orbital elements in more general cases. For example, in the cases of null inclination, hyperbolic or parabolic configurations.

\section{Preliminaries}
We denote in bold every three dimensional vectors and $^\mathsf{T}$ denotes the transpose of a vector with respect to the Euclidean scalar product.

\subsection{Unperturbed Orbit}
In this section, we remind several formulas concerning the orbital elements. We refer to \cite{burns1976} and \cite[Chapter 2]{murray_dermott} for more details. \\

We consider a particle of mass $M_P$ moving in the $r^{-2}$ gravitational field of a fixed point mass $M_S$. The Newton's equation of motion is
\begin{equation}
\frac{d^2 \textbf{r}}{dt^2}=-\frac{\mu}{r^3}\textbf{r} \label{eqmotion}
\end{equation}
where $\mu=G(M_S+M_P)$, $G$ being the universal gravitational constant, $\textbf{r}$ is the position vector from $M_S$ to $M_P$. We denote by $\textbf{v}$ the velocity vector and $\textbf{H}=\textbf{r}\times \textbf{v}$ the angular momentum per unit of mass. Its norm
\begin{equation}
H=r^{2}\frac{d\theta}{dt} \label{eqnormH}
\end{equation}
is conserved with $\theta$ being the position angle measured from some fixed line in the plane. As usual, we choose this line to be the line of nodes (see Figure \ref{fig2}). The total energy per unit of mass is conserved and is defined by
\begin{equation}
E=\frac{1}{2}\textbf{v}^2-\frac{\mu}{r} \label{eqE}.
\end{equation}
The orbit $r$ is function of $\theta$ and is defined in the elliptical case by
\begin{equation}
r=\frac{p}{1+e\cos(\theta-\omega)} .\label{eqr}
\end{equation}

\begin{figure}[ht!]
	\centering
	\includegraphics[width=0.8\textwidth]{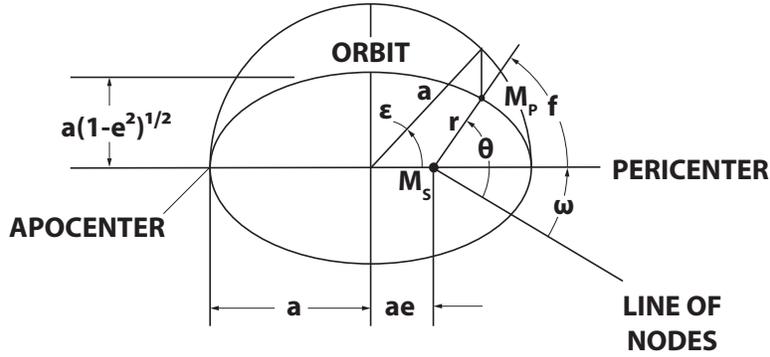}
	\caption{Diagram of the orbit plane of an elliptic orbit, showing the definition of the orbital elements $(a,e,\omega)$, the true anomaly $f$, the eccentric anomaly $\epsilon$ and pericenter location.}
	\label{fig1}
\end{figure}

The quantities $e$ and $\omega$ are constants and determined from the initial conditions. The parameter $p$ is the conic parameter given by
\begin{equation}
p \equiv \frac{H^{2}}{\mu} = a(1-e^{2}) .\label{eqp}
\end{equation}
The right-hand side of (\ref{eqp}) defines $a$ and the argument of the cosine term in (\ref{eqr}) is used to introduce the true anomaly,
\begin{equation}
f\equiv\theta-\omega ,\label{eqf}
\end{equation}
the particle's angular position measured from pericenter (see Figs. \ref{fig1} and \ref{fig2}). An equivalent solution is
\begin{equation}
r=a(1-e\cos\epsilon) \label{eqreps}
\end{equation}
where $\epsilon$ is the eccentric anomaly (see Fig. \ref{fig1}). The true anomaly is related to the eccentric anomaly by
\begin{equation}
\cos \epsilon = \frac{e+\cos f}{1+e\cos f} \quad \text{and} \quad \sin \epsilon = \frac{\sqrt{1-e^2}}{1+e\cos f} \sin f \label{eqsineps}.
\end{equation}
The particle's radial (resp. transverse) velocity is defined as
\begin{equation}
\frac{dr}{dt}=\frac{H}{p}e\sin f \label{eqdotr} \quad \left(\text{resp.} \quad  r\frac{d\theta}{dt}=\frac{H}{p}(1+e\cos f) \right).
\end{equation}

\subsection{Orbital elements}
To describe the particle orbit as a function of time, six constants, are required. These constants are chosen to be the orbital elements. Three orbital elements, $a,e$ and $\omega$, have already been presented. A fourth is needed to completely describe the two-dimensional motion of the particle in the orbital plane. Usually the mean anomaly $M$, related to the Kepler's equation as
\begin{equation}
M =\epsilon-e\sin\epsilon, \label{eqkep}
\end{equation}
is chosen. Using Equations \eqref{eqsineps}, we obtain an equivalent form, 
\begin{equation}
M =\arctan\left(\frac{\sqrt{1-e^2} \sin f}{e+\cos f}\right)-\frac{e\sqrt{1-e^2} \sin f}{1+e\cos f} \label{eqkepbis}.
\end{equation}
The remaining two orbital elements, the inclination $i$ and the longitude of the ascending node $\Omega$, give the orientation of the orbital plane in space as shown in Figure \ref{fig2}. \\

Let $\{\textbf{e}_R,\textbf{e}_T,\textbf{e}_N\}$ being an orthogonal unit vector base where $\textbf{e}_R$ is the normalized radial vector $\textbf{r}$, $\textbf{e}_T$ is transverse to the radial vector in the orbit plane (positive in the direction of motion of the particle) and $\textbf{e}_N$ is normal to the orbit plane in the direction $\textbf{H}$.

\begin{figure}[ht!]
	\centering
	\includegraphics[width=0.8\textwidth]{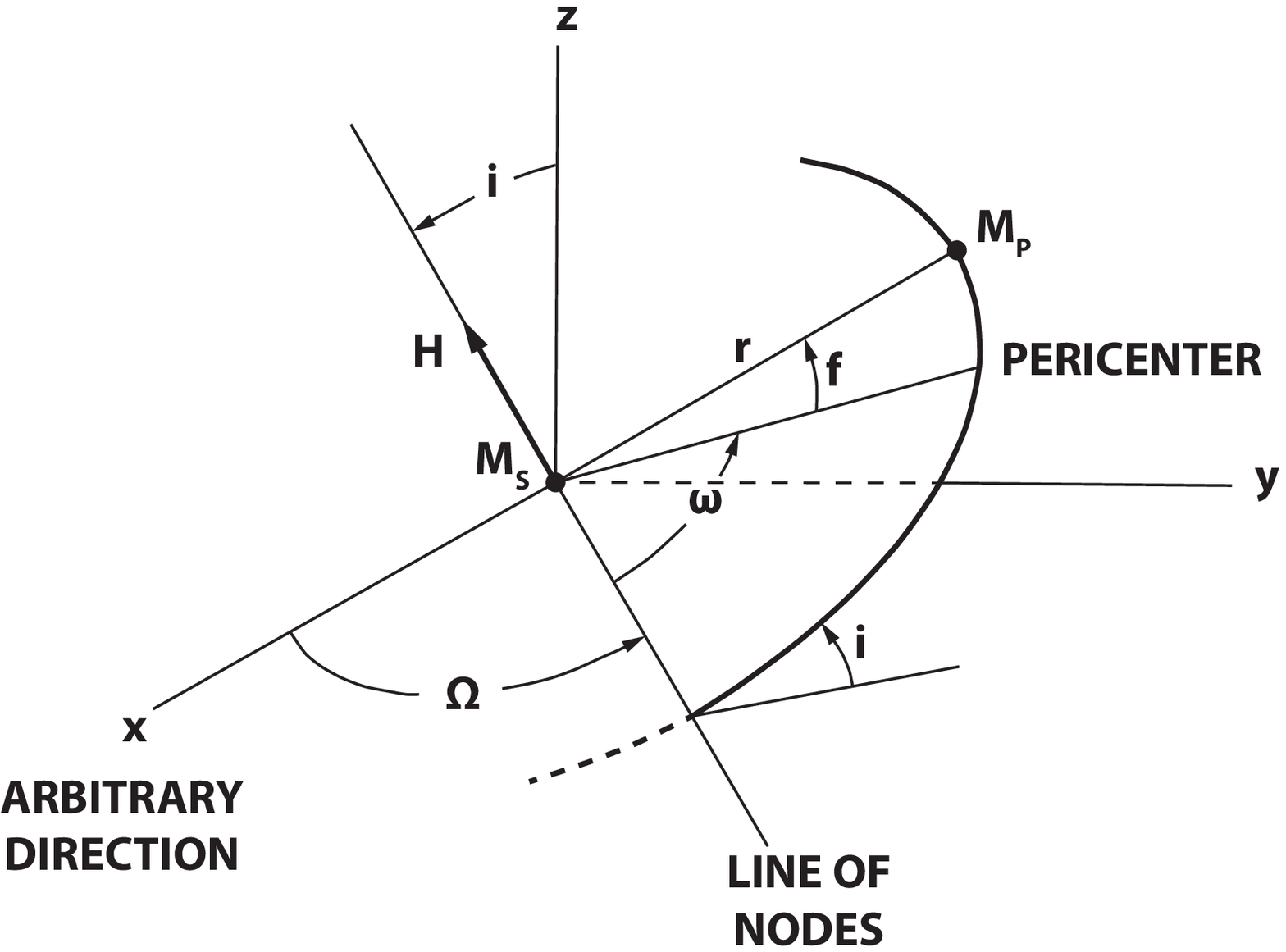}
	\caption{Orbital motion with respect to the reference plane in three dimensional space.}
	\label{fig2}
\end{figure}

\subsection{Energy and angular momentum of the orbit}
We want to express the orbital elements in terms of the orbital energy per unit of mass and angular momentum per unit of mass. We have the well known relations
\begin{equation}
H=\sqrt{\mu a(1-e^{2})} \label{eqHorb}
\end{equation}
and
\begin{equation}
E=-\frac{\mu}{2a} .\label{eqEorb}
\end{equation}
The semi-major axis is only determined by $E$ and the orbital eccentricity is only determined by $E$ and $H$ as
\begin{equation}
e=\sqrt{1+\frac{2 H^2 E}{\mu^2}} . \label{eqe}
\end{equation}
Similarly, as it can be seen from Figure \ref{fig2}, $i$ and $\Omega$ are given by components of the angular momentum vector per unit of mass vector as
\begin{align}
\cos i&=\frac{H_{z}}{H} . \label{eqcosi} \\
\tan\Omega&=-\frac{H_{x}}{H_{y}} , \label{eqtanOmega}
\end{align}
where $H_{x},\ H_{y}$ and $H_{z}$ are the components of \textbf{H} in the inertial reference system attached to $M_S$. Equations (\ref{eqHorb})-(\ref{eqtanOmega}) give four orbital elements in terms of four pieces of information contained in $\textbf{H}$ and $E$.

\section{Perturbed problem}

The problem to be solved is to find the equations governing the time rate of change of the set $(a,e,i,\omega,\Omega,M)$ induced by the action of a stochastic perturbing force $\textbf{F}$.

\subsection{Reminder about stochastic differential equations}

We remind basic properties and definition of stochastic differential equations in the sense of It\^o. We refer to the book \cite{oksendal} for more details and basic properties of the It\^o stochastic calculus.\\

A {\it stochastic differential equation} is formally written (see \cite[Chapter V]{oksendal}) in differential form as  
\begin{equation}
\label{eqdX}
dX_t = \mu (t,X_t)dt+\sigma(t,X_t)dB_t ,
\end{equation}
which corresponds to the stochastic integral equation
\begin{equation}
\label{eqX}
X_t=X_0+\int_0^t \mu (s,X_s)\,ds+\int_0^t \sigma (s,X_s)\,dB_s ,
\end{equation}
where the second integral is an It\^o integral (see \cite[Chapter III]{oksendal}) and $B_t$ is the classical Brownian motion (see \cite[Chapter II, p.7-8]{oksendal}).\\

We now turn to the situation in higher dimensions: Let $\textbf{B}(t)=(B_1(t),\ \ldots,\ B_m(t))^\mathsf{T}$ denote $m$-dimensional Brownian motion. We can form the following $n$ It\^{o} processes
\begin{align}
\left\{\begin{array}{ccccc}
&dX_{1}(t)=&\bar{X}_{1}(t)dt+&\tilde{X}_{11}dB_1(t)+\cdots+&\tilde{X}_{1m}dB_m(t)\\
&\vdots &\vdots & \vdots &\vdots\\
&dX_{n}(t)=&\bar{X}_{n}(t)dt+&\tilde{X}_{n1}dB_1(t)+\cdots+&\tilde{X}_{nm}dB_m(t)
\end{array}\right. 
\end{align}
for $(1\leq i\leq n,\ 1\leq j\leq m)$. Or, in matrix notation simply
\begin{equation}
d\textbf{X}(t)=\bar{\textbf{X}}(t)dt+\tilde{\textbf{X}}(t)\cdot d\textbf{B}(t) ,
\end{equation}
where
\begin{align*}
\textbf{X}=\left(\begin{array}{c}
X_{1}\\
\vdots \\
X_{n}
\end{array}\right),
\ \bar{\textbf{X}}=\left(\begin{array}{c}
\bar{X}_{1}\\
\vdots \\
\bar{X}_{n}
\end{array}\right),
\ \tilde{\textbf{X}}=\left(\begin{array}{ccc}
\tilde{X}_{11} & \cdots & \tilde{X}_{1m}\\
\vdots  &  & \vdots \\
\tilde{X}_{n1} & \cdots & \tilde{X}_{nm}
\end{array}\right),
\ d\textbf{B}(t)=\left(\begin{array}{c}
dB_1(t)\\
\vdots \\
dB_m(t)
\end{array}\right) .
\end{align*}
Such a process $X(t)$ is called an $n$-dimensional It\^{o} process (or just an It\^{o} process). An important tool to study functions which depend of stochastic processs is the general It\^{o} formula. Let $\textbf{X}(t)$ be an $n$-dimensional It\^{o} process as above and let $g(t,x)=(g_{1}(t,x),\ \ldots \ ,g_{p}(t,x))^\mathsf{T}$ be a $C^{2}$ map from $\mathbb{R}^+ \times \mathbb{R}^{n}$ into $\mathbb{R}^{p}$. Then the process $\textbf{Y}(t)=g(t,\textbf{X}(t))$ is again an It\^{o} process, whose component number $k,Y_{k}$, is given by
\begin{equation}
dY_{k}=\displaystyle \frac{\partial g_{k}}{\partial t}(t,\textbf{X})dt+\sum_{i}\frac{\partial g_{k}}{\partial x_{i}}(t,\textbf{X})dX_{i}+\frac{1}{2}\sum_{i,j}\frac{\partial^{2}g_{k}}{\partial x_{i}\partial x_{j}}(t,\textbf{X})dX_{i}dX_{j} ,
\end{equation}
where $dB_{t,i}\,dB_{t,j}=\delta_{ij}\, dt$ and $dB_{t,i}\,dt=dt\,dB_{t,i}=0$. Denoting $\tilde{\textbf{X}}_p=(\tilde{X}_{p1},\cdots, \tilde{X}_{pm})^\mathsf{T}$, for all $1 \leq p \leq n$, $dY_k$ can be written as
\begin{equation}
dY_{k}=\displaystyle \bigg[\frac{\partial g_{k}}{\partial t}(t,\textbf{X}) +\frac{1}{2}\sum_{i,j}\frac{\partial^{2}g_{k}}{\partial x_{i}\partial x_{j}}(t,\ X)\tilde{\textbf{X}}_i \cdot \tilde{\textbf{X}}_j \bigg]dt+\sum_{i}\frac{\partial g_{k}}{\partial x_{i}}(t,\textbf{X})dX_{i} .
\end{equation}

\subsection{Equations of perturbed motion}
In the following, for notation convenience, we omit the dependence for each process and the Brownian motion. \\

First, we write in the differential form equations of motion to be coherent with the formulation of stochastic differential equations. We recall that \textbf{r} is the vector position from $M_S$ to $M_P$ and \textbf{v} is the velocity vector. Thus, we have
\begin{align}
d\textbf{r} &= \textbf{v}dt, \label{eqdvecr} \\
d\textbf{v} &= -\frac{\mu}{r^3}\textbf{r} \ dt + d\textbf{v}_P, \label{eqdvecv}
\end{align}
where $d\textbf{v}_P$ corresponds to the perturbing acceleration induced by the perturbing force $\textbf{F}$. In $\{\textbf{e}_R,\textbf{e}_T,\textbf{e}_N\}$, the position vector is $\textbf{r} = r \textbf{e}_R$. Then, its variation is given by
\begin{align}
d\textbf{r} &= dr \textbf{e}_R + r d\theta \textbf{e}_T .
\end{align}
Let $v$ be the radial velocity and $w$ be the transverse velocity defined by
\begin{equation}
dr=vdt \quad \text{and} \quad d\theta=wdt .
\end{equation}
Thus, the variation of the position vector is finally given by
\begin{align}
d\textbf{r} &= \left(v \textbf{e}_R + r w \textbf{e}_T \right)dt
\end{align}
and we identify the velocity vector $\textbf{v}$ as
\begin{align}
\textbf{v}=v \textbf{e}_R + r w \textbf{e}_T.
\end{align}
It follows the variation of the velocity vector is given by
\begin{equation}
d\textbf{v} = (dv-rw^2dt)\textbf{e}_R + (2vwdt+rdw)\textbf{e}_T .
\end{equation}
In order to get the expression of the radial and transverse acceleration, we make precise the expression of the perturbing acceleration $d\textbf{v}_P$. \\

\noindent\textbf{Stochastic perturbing acceleration:}
\emph{Let $\textbf{B}$ be a $m$-dimensional Brownian motion. The stochastic perturbing acceleration is defined as
\begin{equation}
d\textbf{v}_{P} = \bar{\textbf{v}}_P \ dt + \tilde{\textbf{v}}_P \cdot d\textbf{B} \label{model_sto}
\end{equation}
where $\bar{\textbf{v}}_P =(\bar{R},\bar{T},\bar{N})^\mathsf{T}$ is, in our problem, the deterministic part of the perturbation 
\begin{equation*}
\tilde{\textbf{v}}_P=\left(\begin{matrix}
\tRR_1 & \tRR_2 & \cdots & \tRR_m \\ 
\tTT_1 & \tTT_2 & \cdots & \tTT_m \\ 
\tNN_1 & \tNN_2 & \cdots & \tNN_m
\end{matrix}\right)
\end{equation*}
is the purely stochastic part of the perturbation.}  \\

In what follows, we denote $\tR$, $\tT$ and $\tN$ the rows of $\tilde{\textbf{v}}_P$. We also simplify the notation for the scalar product of a vector $\mathbf{u}$ with itself, $\mathbf{u}\cdot \mathbf{u}$ as $\mathbf{u}^2$. \\

Using (\ref{eqdvecv}) and the expression of the stochastic perturbing acceleration \eqref{model_sto}, we obtain the final expression of the radial and transverse accelerations written as
\begin{align}
dv &= \left( rw^2-\frac{\mu}{r^2} + \bar{R} \right)dt + \tR\cdot d\textbf{B},\label{dv}\\
dw &= \left(-\frac{2vw}{r} + \frac{\bar{T}}{r} \right)dt + \frac{\tT}{r}\cdot d\textbf{B} \label{dw}.
\end{align}

In first consequences, we obtain the variations of the angular momentum and the energy as follows:

\begin{lemma}
The variation of the angular momentum $H$ is given by
\begin{equation}
dH =\frac{a(1-e^2)}{1+e\cos f} \bar{T} dt + \frac{a(1-e^2)}{1+e\cos f} \tT\cdot d\textbf{B}, \label{dH}
\end{equation}
and the variation of the energy $E$ is given by
\begin{align}
dE = &\left[\sqrt{\frac{\mu}{a(1-e^2)}} \left(e\sin f \bar{R} +(1+e\cos f)\bar{T}\right)+\frac{\tR^2+\tT^2}{2}\right]dt \nonumber \\ 
&+ \sqrt{\frac{\mu}{a(1-e^2)}} \left( e\sin f \tR +(1+e\cos f)\tT \right) \cdot d\textbf{B}  \label{dE}
\end{align}
\end{lemma}

\begin{proof}
The computation is straightforward. Using It\^o's formula, we obtain
\begin{align}
dH =& r \bar{T} dt + r \tT\cdot d\textbf{B}, \\
dE = &\left( v\bar{R} +rw\bar{T}+\frac{\tR^2}{2}+\frac{\tT^2}{2}\right)dt + \left(v\tR + rw \tT\right) \cdot d\textbf{B}.
\end{align}
Using the formula from \eqref{eqr} to \eqref{eqsineps}, we obtain the result.
\end{proof}

{\bf Remark:}
In Equation \eqref{de}, the scalar product $\tR^2$ and $\tT^2$ are exactly the \emph{supplementary terms} obtained with the It\^o formula. Contrary to the classical derivation of the Gauss equations, the stochastic nature of the perturbation induces these extra terms. In consequence, it will bring new terms in the variation of the orbital elements related to the energy. The apparition of these new terms are exactly the reason and the need of a new set of Gauss equations.

\section{Stochastic Gauss Equations in terms of $a,e,i,\Omega,\omega,M$}

In this section, we obtain the equations governing the variation of the orbital elements $a,e,i,\Omega,\omega$ and $M$ induced by the stochastic perturbing acceleration \eqref{model_sto}. All the proofs are given in Appendix.

\begin{lemma}[The Semi-major axis $a$]
The variation of the semi-major axis $a$ is given by
\begin{align}
da=&\bigg[ \frac{2 a^{3/2}}{\sqrt{\mu (1-e^2)}}\left(e \sin f \bar{R} + (1+e\cos f) \bar{T}  \right) \nonumber \\
&+ \frac{a^2}{\mu} \left( \left(1+\frac{4e^2 \sin^2 f}{1-e^2}\right)\tR^2 + \left(1+\frac{4(1+e\cos f)^2}{1-e^2}\right)\tT^2 \right) \nonumber \\
&+\frac{8a^2}{\mu (1-e^2)}e\sin f(1+e\cos f)\tR\cdot\tT \bigg]dt \nonumber \\
&+ \frac{2 a^{3/2}}{\sqrt{\mu (1-e^2)}}\left(e \sin f \tR + (1+e\cos f) \tT  \right)\cdot d\textbf{B} . \label{da}
\end{align}
\end{lemma}

The proof is given in Section \ref{dem_da}.

\begin{lemma}[The Eccentricity $e$]
The variation of the eccentricity $e$ is given by
\begin{align}
de 	= &\bigg[ \sqrt{\frac{a(1-e^2)}{\mu}}\left(\sin f \bar{R} + (\cos f + \frac{e+\cos f}{1+e\cos f})\bar{T} \right) +  \frac{a(1-e^2)\cos^2 f}{2e\mu}\tR^2  \nonumber \\
&+\frac{a(1-e^2)}{\mu e}\left(2-\frac{\cos f}{2}\left(\frac{2+e\cos f}{1+e \cos f}\right) \left(\cos f + \frac{e+\cos f}{1+e\cos f}\right)\right)\tT^2 \nonumber \\
&+ \frac{a(1-e^2)}{\mu e (1+e\cos f)}(e\sin^3 f - \sin 2f) \tR \cdot \tT\bigg]dt \nonumber \\
&+\sqrt{\frac{a(1-e^2)}{\mu}}\left(\sin f \tR + (\cos f + \frac{e+\cos f}{1+e\cos f})\tT \right)\cdot d\textbf{B} . \label{de}
\end{align}
\end{lemma}
 
 The proof is given in Section \ref{dem_de}.
 
\begin{lemma}[The Inclination $i$ and the ascending node $\Omega$]
The variation of the inclination $i$ is given by
\begin{align}
di = &\bigg[\sqrt{\frac{a(1-e^2)}{\mu}}\frac{\cos (f+\omega)}{(1+e\cos f)} \bN \nonumber \\
& - \frac{a \left(1-e^2\right)}{\mu (1+ e\cos f)^2}\cos(f+\omega)\left( \frac{\cot i \cos(f+\omega)}{2} \tN^2 + \tT \cdot \tN \right)\bigg]dt \nonumber \\
&+ \sqrt{\frac{a(1-e^2)}{\mu}}\frac{\cos (f+\omega)}{(1+e\cos f)} \tN \cdot d\textbf{B} . \label{di} 
\end{align}
and the variation of the ascending node $\Omega$ is given by
\begin{align}
d\Omega = &\bigg[\sqrt{\frac{a(1-e^2)}{\mu}} \frac{\sin(f+\omega)}{\sin i (1+e\cos f)} \bN \nonumber \\
&-  \frac{a \left(1-e^2\right)}{\mu (1+ e\cos f)^2}\frac{\sin (f+\omega)}{\sin i}\left(\cos(f+\omega) \cot i \, \tN^2 + \tT \cdot \tN \right) \bigg]dt  \nonumber \\
&+ \sqrt{\frac{a(1-e^2)}{\mu}} \frac{\sin(f+\omega)}{\sin i (1+e\cos f)}\tN \cdot d\textbf{B} . \label{dOmega}
\end{align}

\end{lemma}

The proof is given in Section \ref{dem_didW}.

\begin{lemma}[The pericenter $\omega$]
The variation of the pericenter $\omega$ is given by
\begin{align}
d\omega &=\bigg[ \sqrt{\frac{a(1-e^2)}{\mu}}\left( -\frac{\cos f}{e} \bR + \frac{\sin f}{e} \left(\frac{2+e\cos f}{1+e\cos f} \right)\bT \right)  \nonumber \\
&+ \frac{a(1-e^2)}{\mu e^2}\bigg( \frac{\sin 2f}{2} \tR^2 - \left(e+\cos f(2+e\cos f)^2 \right)\frac{\sin f}{(1+e\cos f)^2} \tT^2 \nonumber \\
&+ \left(\frac{2+e\cos f}{1+e\cos f}\right)\cos 2f \tR \cdot \tT \bigg) +\frac{a(1-e^2)\sin(2(f+\omega))}{4\mu(1+e\cos f)^2} \tN^2\bigg]dt \nonumber \\
&+\sqrt{\frac{a(1-e^2)}{\mu}}\left( -\frac{\cos f}{e} \tR + \frac{\sin f}{e} \left(\frac{2+e\cos f}{1+e\cos f} \right)\tT \right) \cdot d\textbf{B}  -\cos i \ d\Omega\label{domega}.
\end{align}
\end{lemma}

The proof is given in Section \ref{dem_dw}.

\begin{lemma}[The mean anomaly $M$]
The variation of the mean anomaly $M$ is given by
\begin{align}
dM = &\bigg[\sqrt{\frac{\mu }{a^3}} -2\sqrt{\frac{a}{\mu}}\frac{(1-e^2)}{1+e\cos f} \bR  + \frac{a(1-e^2)^{3/2}\sin(2(f+\omega))}{4\mu(1+e\cos f)^2}\tN^2 \nonumber\\
&+\frac{a \sqrt{1-e^2}\sin f}{2 \mu  (1+e \cos f)}\bigg(\left(2e -\cos f(1+e \cos f)\right) \tR^2 \nonumber \\
&+\frac{ (2+e \cos f)}{(1+e\cos f)} (\cos f (2+e \cos f)+e)\tT^2 \nonumber \\
&+2\sin f (2+e \cos f)\tR\cdot\tT\bigg)\bigg]dt \nonumber \\
&-2\sqrt{\frac{a}{\mu}}\frac{(1-e^2)}{1+e\cos f} \tR \cdot d\textbf{B}-\sqrt{1-e^2}\left(d\omega+\cos i d\Omega\right) \label{dM}
\end{align}
\end{lemma}

The proof is given in Section \ref{dem_dM}. \\

Together, the six equations \eqref{da}, \eqref{de}, \eqref{di}, \eqref{dOmega}, \eqref{domega} and \eqref{dM} are called the \emph{stochastic Gauss Equations} in terms of orbital elements $a,e,i,\Omega,\omega,M$.

\section{An example with numerical simulations: Motion of a satellite undergoing stochastic dissipation}
In this section we give an example of a stochastic perturbation of the two-body problem in order to illustrate the stochastic Gauss equations. \\

We consider the following perturbed problem
\begin{equation}
\ddot{\textbf{r}}=-\frac{\mu}{r^3}\textbf{r} + (\alpha_0+\alpha W_1)\frac{ \textbf{v}}{\|\textbf{v}\|} + (\beta_0+\beta W_2) \frac{\textbf{H}}{\|\textbf{H}\|} \label{pb_pert}
\end{equation}
where $\alpha_0,\beta_0,\alpha$ and $\beta$ are real constants, $W_1$ and $W_2$ are two white noises. \\

This perturbed problem can bee seen, for example, as a satellite moving around the Earth which undergoes atmospheric dragging and with normal perturbation. Such perturbations can be induced by the Earth's atmosphere, the Earth's magnetic field fluctuations, the radiation pressure, thermic dissipation etc. The two white noises model the highly fluctuations induced by the phenomena considered. Such considerations are the same in the approach of Sagirow's satellite problem (see \cite{sagirow}). By definition of the vector \textbf{v} and \textbf{H}, we have
\begin{equation}
(\alpha_0+\alpha W_1)\frac{ \textbf{v}}{\|\textbf{v}\|} = \frac{(\alpha_0+\alpha W_1)}{\sqrt{1+e^2+2e\cos f}}\left(e\sin f \textbf{e}_R + (1+e\cos f) \textbf{e}_T\right)
\end{equation}
and
\begin{equation}
(\beta_0+\beta W_2) \frac{\textbf{H}}{\|\textbf{H}\|} = (\beta_0+\beta W_2) \textbf{e}_N.
\end{equation}
Thus, the perturbed acceleration $\displaystyle \frac{d\textbf{v}_P}{dt}$ can be written as
\begin{equation}
\frac{d\textbf{v}_P}{dt}=\left(\begin{array}{c} \frac{\alpha_0 e\sin f}{\sqrt{1+e^2+2e\cos f}} \\ \frac{\alpha_0(1+e\cos f)}{\sqrt{1+e^2+2e\cos f}}\\ \beta_0 \end{array}\right) + \left(\begin{array}{c} \frac{\alpha e\sin f}{\sqrt{1+e^2+2e\cos f}} W_1 \\ \frac{\alpha (1+e\cos f)}{\sqrt{1+e^2+2e\cos f}} W_1 \\ \beta  W_2\end{array}\right) 
\end{equation}
where the vectors are expressed in the basis $\{\textbf{e}_R,\textbf{e}_T,\textbf{e}_N\}$. Assuming $W_1$ and $W_2$ are the components of a two dimensional white noise $\textbf{W}=(W_1,W_2)^\mathsf{T}$ then, the It\^o's interpretation of white noises leads to the following stochastic differential equations for $\textbf{v}_P$:
\begin{equation}
d\textbf{v}_P = \left(\begin{array}{c}\frac{\alpha_0 e\sin f}{\sqrt{1+e^2+2e\cos f}} \\ \frac{\alpha_0 (1+e\cos f)}{\sqrt{1+e^2+2e\cos f}}\\ \beta_0 \end{array}\right)dt + \left(\begin{matrix} \frac{\alpha e\sin f}{\sqrt{1+e^2+2e\cos f}} &  0 \\   \frac{\alpha (1+e\cos f) }{\sqrt{1+e^2+2e\cos f}}&  0 \\ 0 &  \beta \end{matrix}\right) \cdot d\textbf{B}
\end{equation}
where $\mathbf{B}=(B_1,B_2)^\mathsf{T}$ is a two dimensional Brownian motion. The only non-vanishing products of $\tR$, $\tT$ and $\tN$ are
$\displaystyle \tR^2=\frac{\alpha^2 e^2\sin^2 f}{1+e^2+2e\cos f}$, $\displaystyle \tT^2=\frac{\alpha^2 (1+e\cos f)^2}{1+e^2+2e\cos f}$, $\displaystyle \tN^2=\beta^2$ and $\displaystyle \tR\cdot\tT=\frac{\alpha^2 e\sin f(1+e\cos f)}{1+e^2+2e\cos f}$. \\

In order to study the stochastic Gauss equations associated to this problem, we perform numerical simulations. These simulations are done over a period $T=50$ with a time step of $h=10^{-2}$, using a stochastic weak order two method given in \cite[Chapter 5, Equation 2.1]{kloeden2} and implemented in a FORTRAN program. For a review of numerical simulations of stochastic differential equations, we refer to \cite{higham} and \cite{kloeden2}. We also refer to \cite{cpp} and \cite{bcp} for other examples of simulations of stochastic perturbations. The distance and time units are chosen to be the canonical units AU and TU. In that case $\mu=1$ (see \cite{bate71}). The initial conditions for the motion are chosen such that at time $t=0$, the orbiting body is in an elliptical configuration with $r=1 \ \mathrm{AU}, \theta= 1 \ \mathrm{rad}, v=0.01 \ \mathrm{AU/TU}$ and $w=1.1 \ \mathrm{rad/TU}$.\\

We decompose the problem in two cases: a first with only the deterministic part and a second, with the deterministic and the stochastic part. In all the orbital elements figures, we plot in green their unperturbed value and in red their perturbed one. \\

\noindent\textbf{First case}: $\alpha=\beta=0$ and $\alpha_0=-2\times 10^{-2}, \beta_0=10^{-2}$. \\

We display in Figure \ref{cas_perturbe_det} the perturbed two-body motion in that case with two different views. In Figure \ref{gauss_det}, we display the variations of $a,e,i,\Omega$ and $\omega$. In that case, it is known (see for example \cite{mavraganis} and references therein) the orbit is spiraling in its orbit plane. The perturbation due to $\beta_0$ induces a rotation of the orbit plane. As we can see, the eccentricity increase with decaying oscillations which make the osculating orbit tending to a more and more elongate ellipse but with its major axis decreasing. This is clearly the effect of the dissipation.\\

\begin{figure}[ht!]
	\begin{center}
		\subfloat{
			\includegraphics[width=0.5\textwidth]{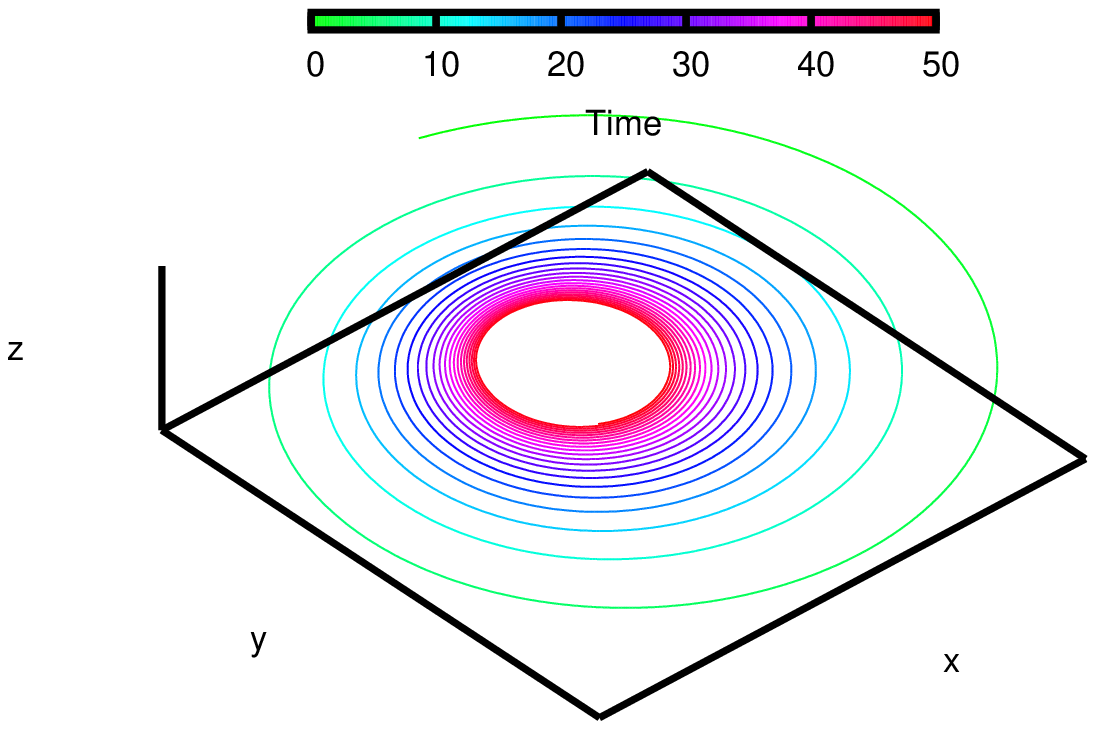}
		}
		\subfloat{
			\includegraphics[width=0.5\textwidth]{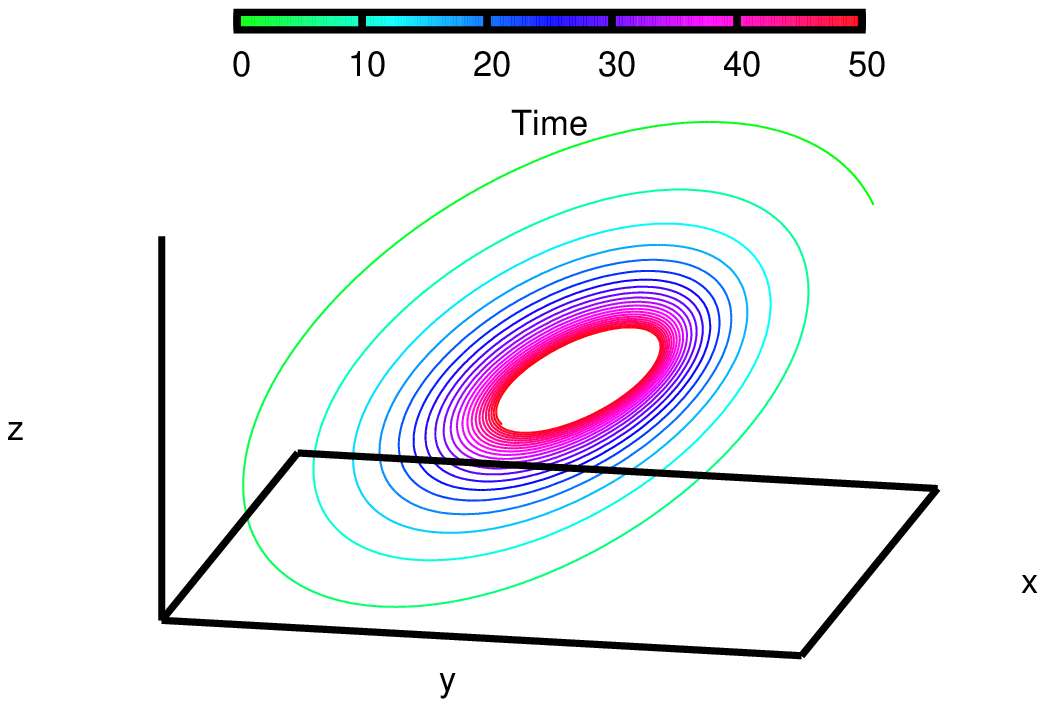}
		}
		\caption{Numerical simulations of the first case with two different views.}
		\label{cas_perturbe_det}
	\end{center}
\end{figure}

\begin{figure}[ht!]
	\resizebox{\textwidth}{!}{
		\begin{tabular}{rr}
			\subfloat{\includegraphics{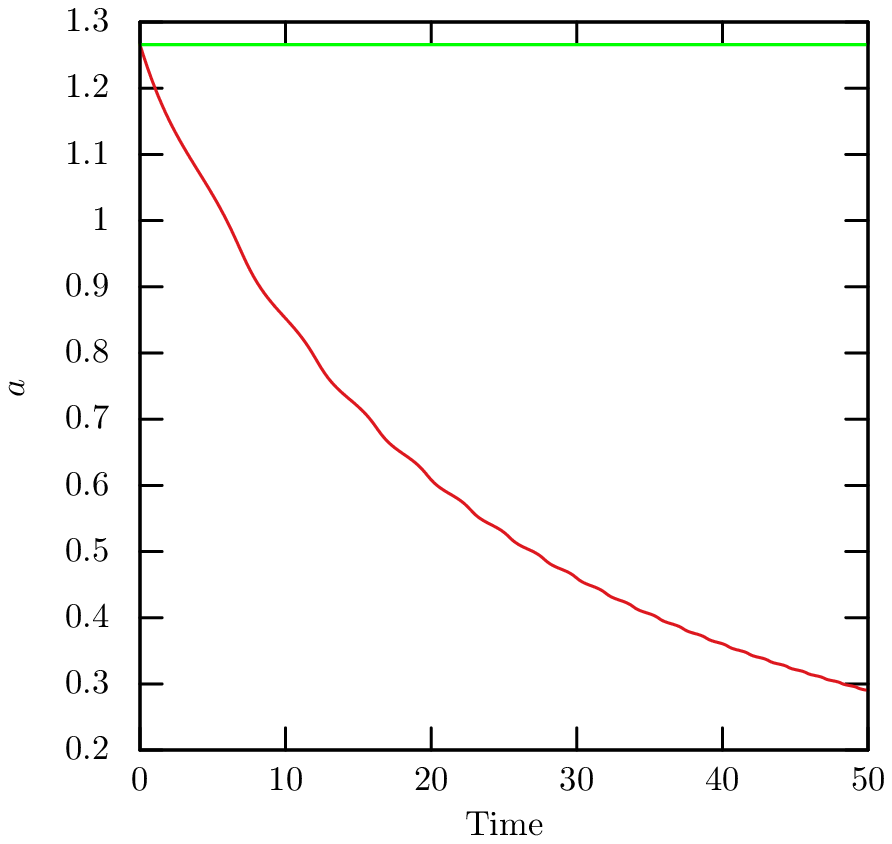}}& \subfloat{\includegraphics{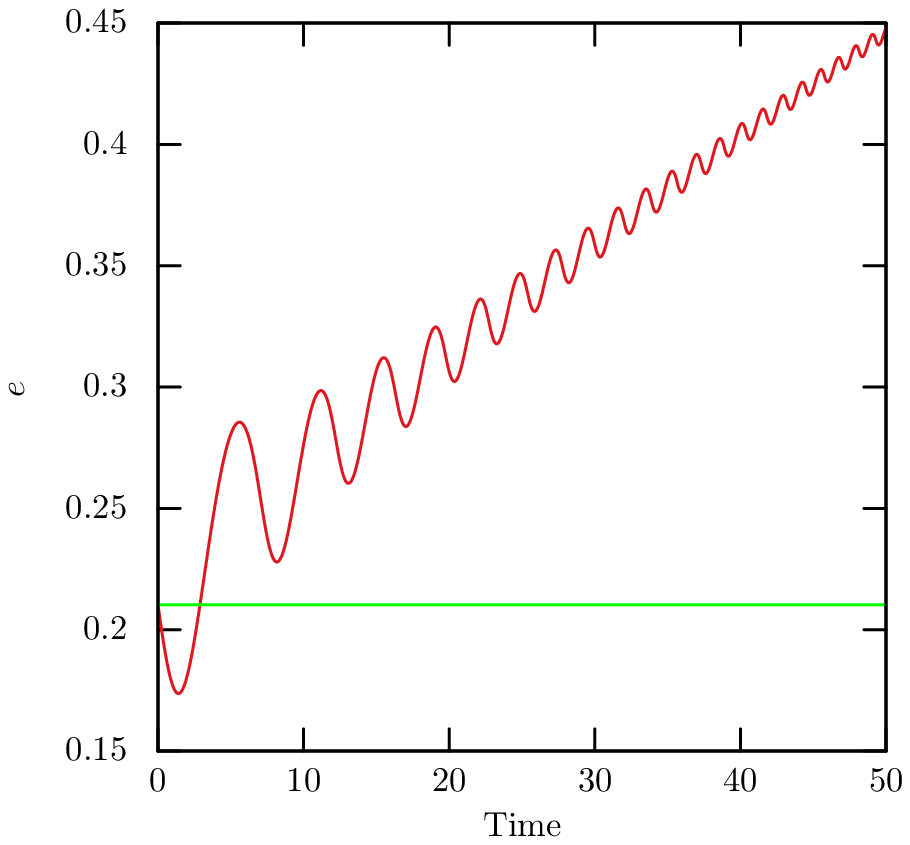}} \\
			
			\subfloat{\includegraphics{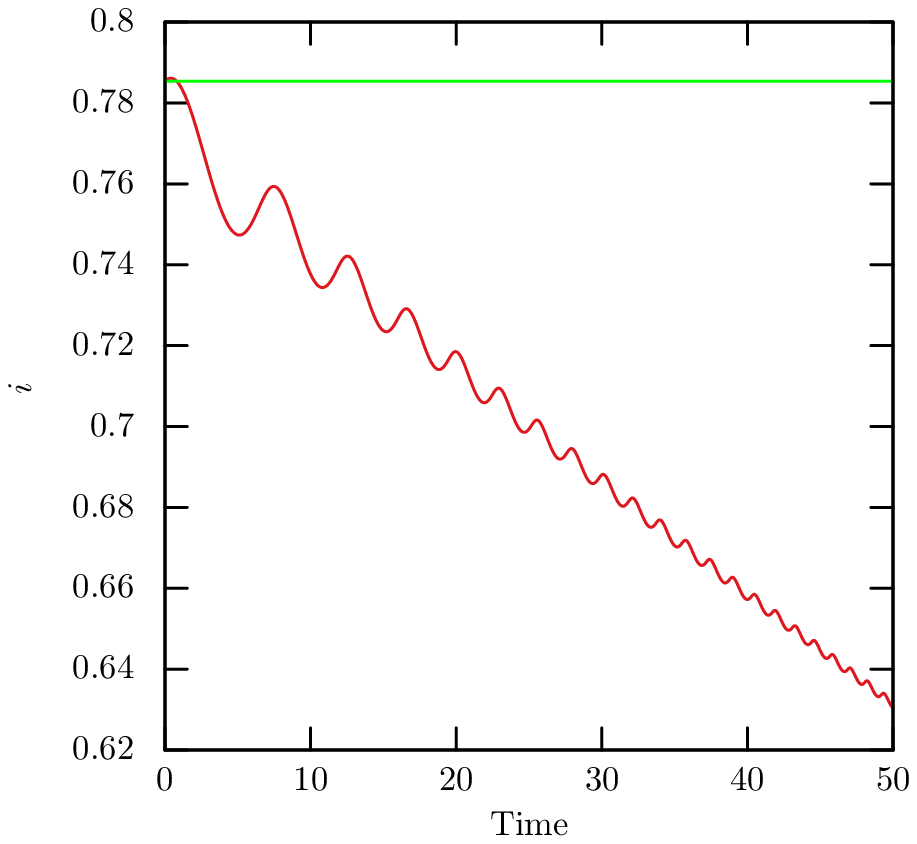}}& \subfloat{\includegraphics{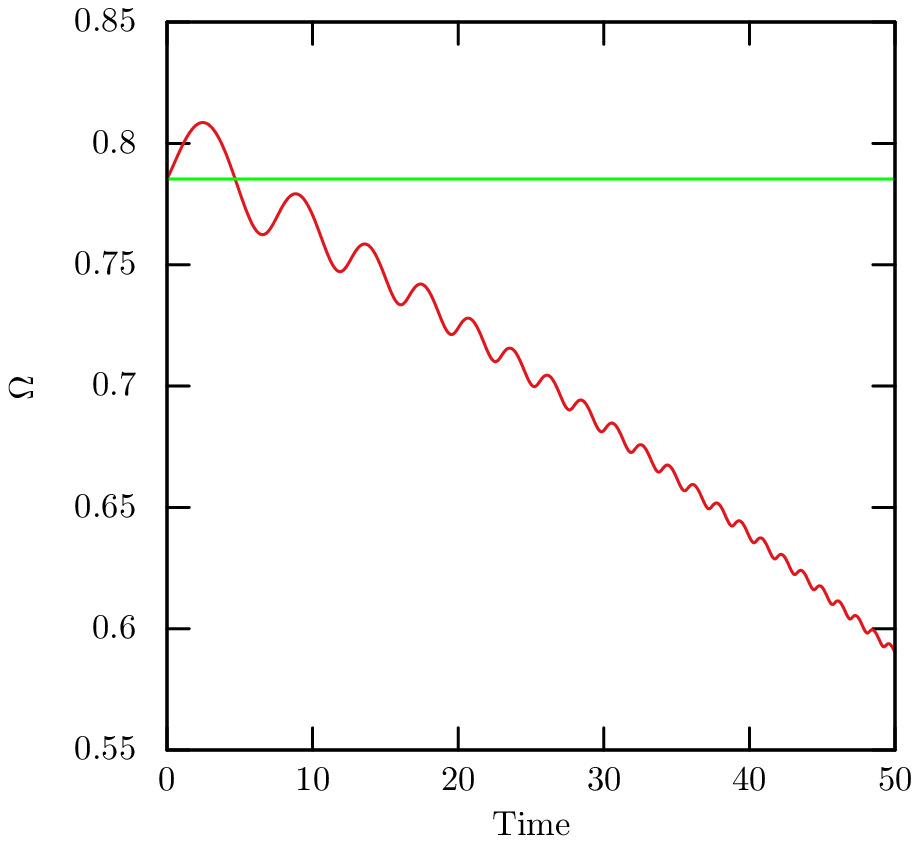}}\\
			
			\multicolumn{2}{c}{\subfloat{\includegraphics{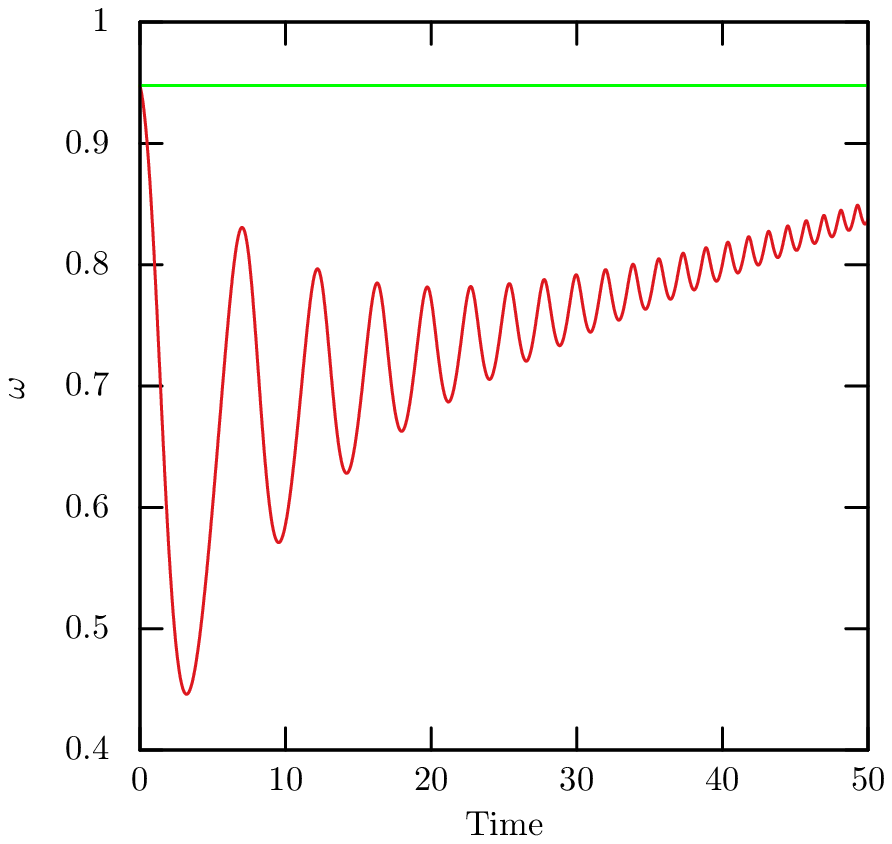}}} \\
		\end{tabular}
	}	
	\centering\caption{First case: Orbital elements.}
	\label{gauss_det}
\end{figure}

\noindent \textbf{Second case}: $\alpha_0=\alpha=-2\times 10^{-2}$ and $\beta_0=\beta=10^{-2}$.\\

The stochastic nature of the perturbation allows multiple realizations. In consequence, a lot of behavior for the motion can exist. We display in Figure \ref{cas_perturbe}, two examples of this case. In order to find the mean behavior of the orbital elements, in the probabilistic sense, we compute the expectation of the orbital elements using a Monte-Carlo method with $10^5$ realizations of Brownian motion. The expectations of the variations of $a,e,i,\Omega$ and $\omega$ obtained with the stochastic Gauss equations are given in Figure \ref{Egauss_detsto}. As we can see with these choices of coefficients, around $T=10$, the stochastic component of the perturbation begins to annihilate the periodic variations of the orbital elements, notably for $e,i,\Omega$ and $\omega$. Moreover, we can see that the orbital elements varying faster. Notably for the eccentricity and the pericenter, the stochastic part makes them drifting quickly than the deterministic case. \\

Even if the coefficient of the deterministic and the purely stochastic part are the same, only the square of the stochastic part remains due to the fact that It\^o's integral, and more precisely, the integral along the variation of the Brownian motion, vanishes in expectation (see \cite[Theorem 3.7, p.22]{oksendal}). In consequence, even if the deterministic part is more important than the stochastic one in term of magnitude, the purely stochastic part induces a non negligible effect on the dynamics. \\

\begin{figure}[ht!]
	\begin{center}
		\subfloat{
			\includegraphics[width=0.5\textwidth]{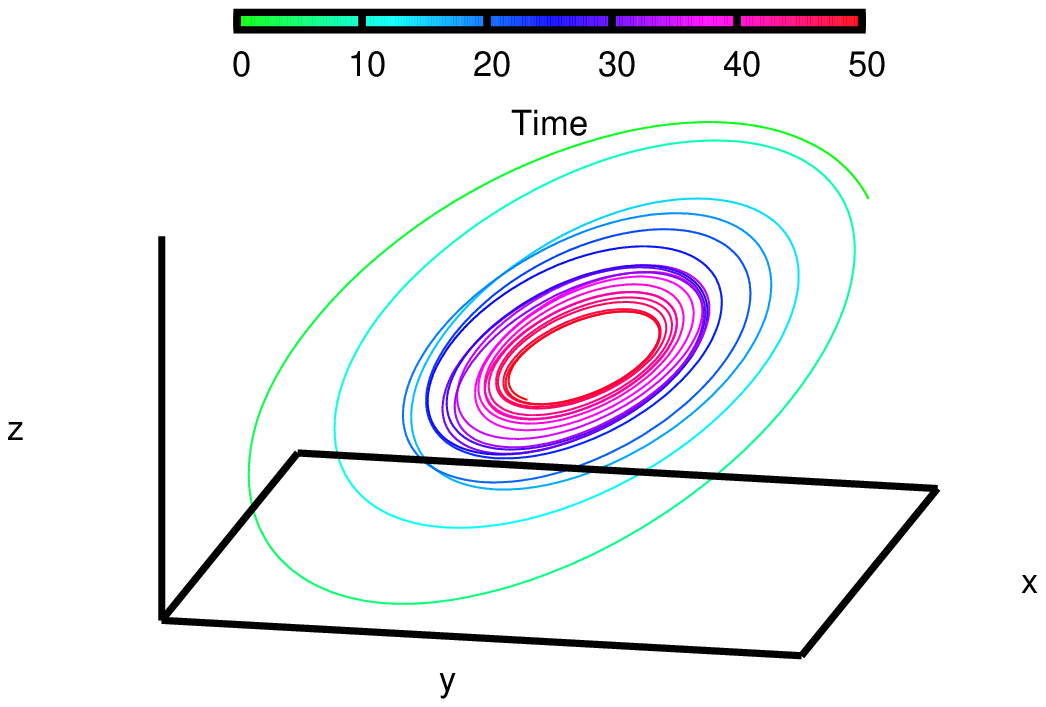}
		}
		\subfloat{
			\includegraphics[width=0.5\textwidth]{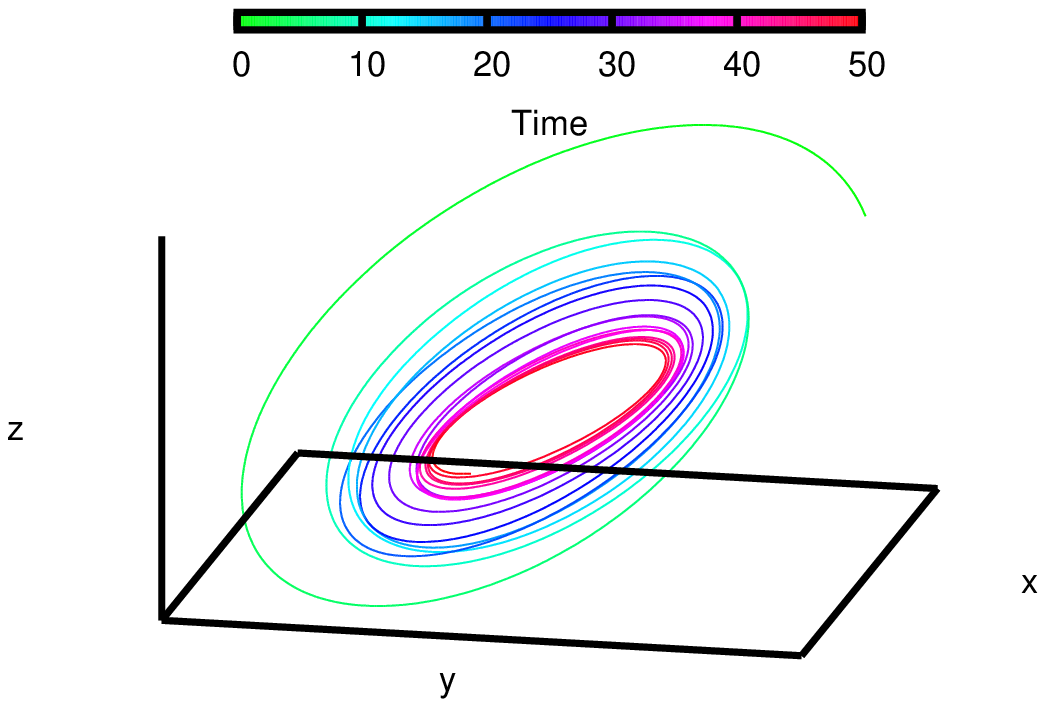}
		}
		\caption{Numerical simulations of two examples of the second case.}
		\label{cas_perturbe}
	\end{center}
\end{figure}

\begin{figure}[ht!]
		\resizebox{\textwidth}{!}{
			\begin{tabular}{rr}
				\subfloat{\includegraphics{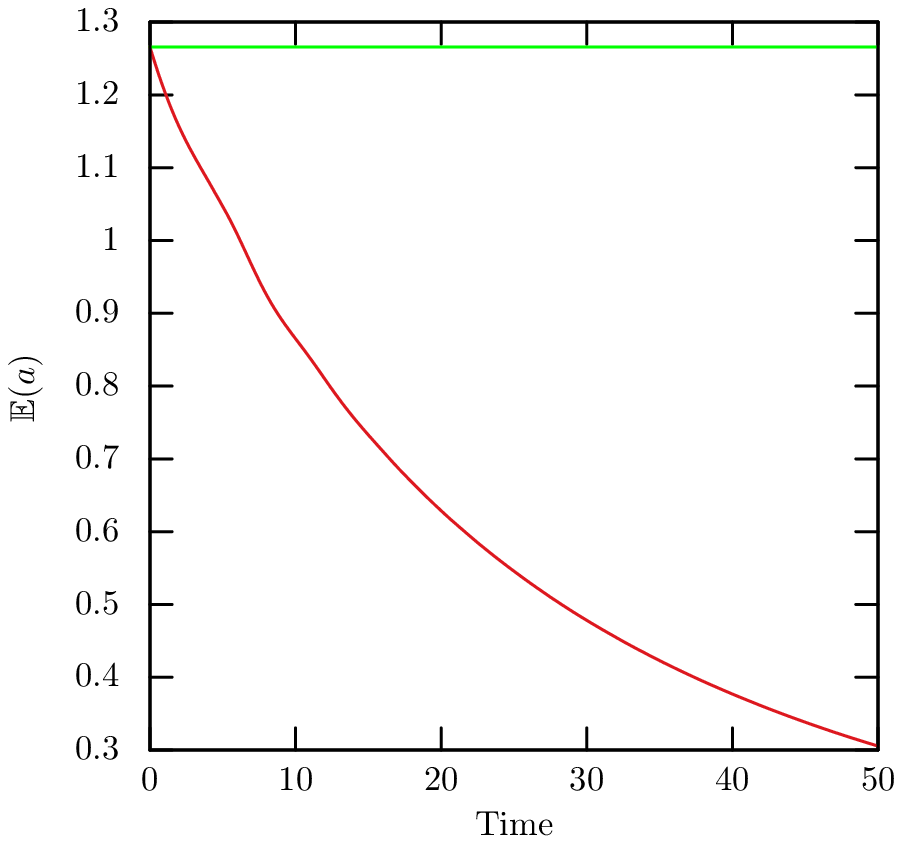}}&	\subfloat{\includegraphics{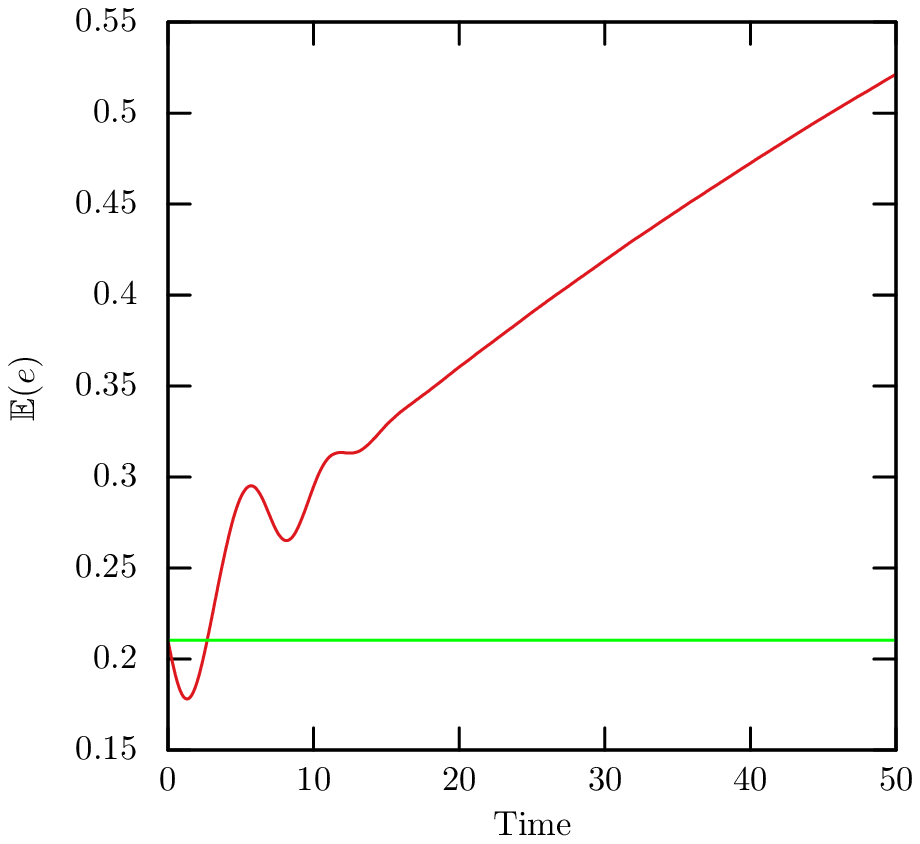}} \\
				
				\subfloat{\includegraphics{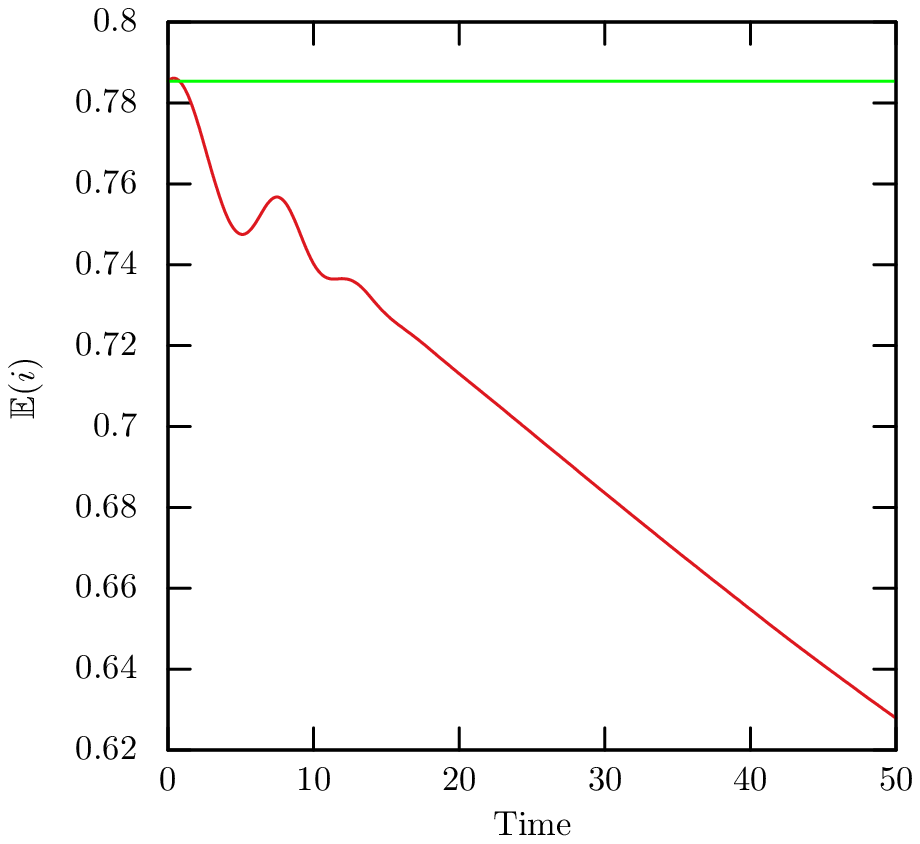}}& \subfloat{\includegraphics{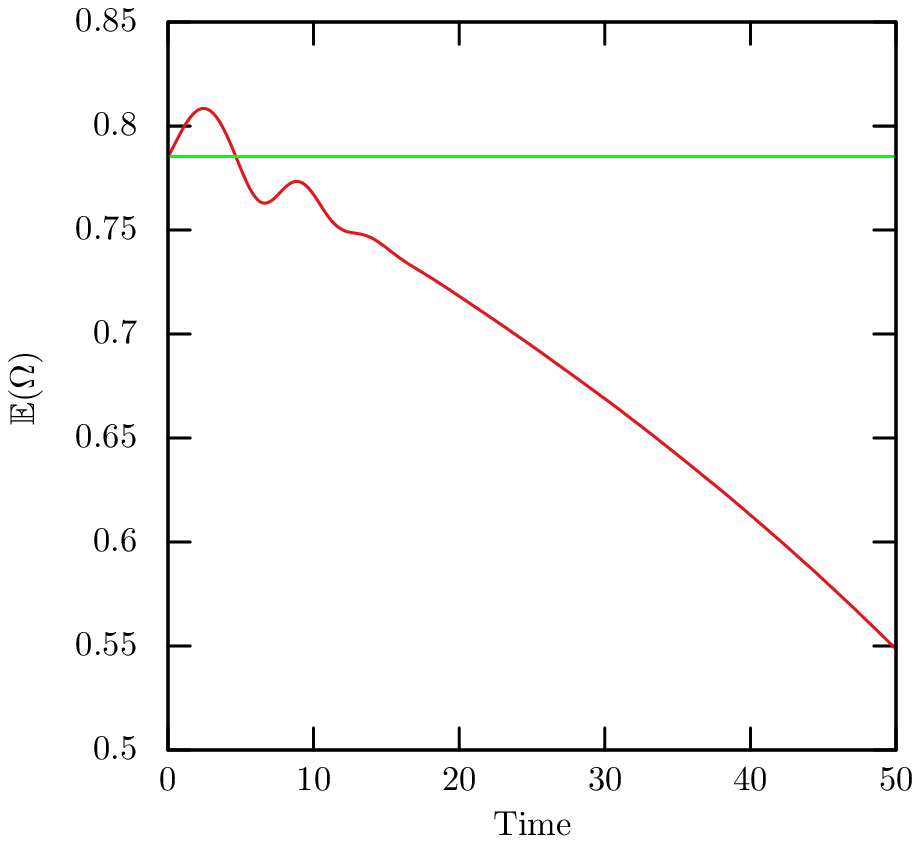}}\\
				
				\multicolumn{2}{c}{\subfloat{\includegraphics{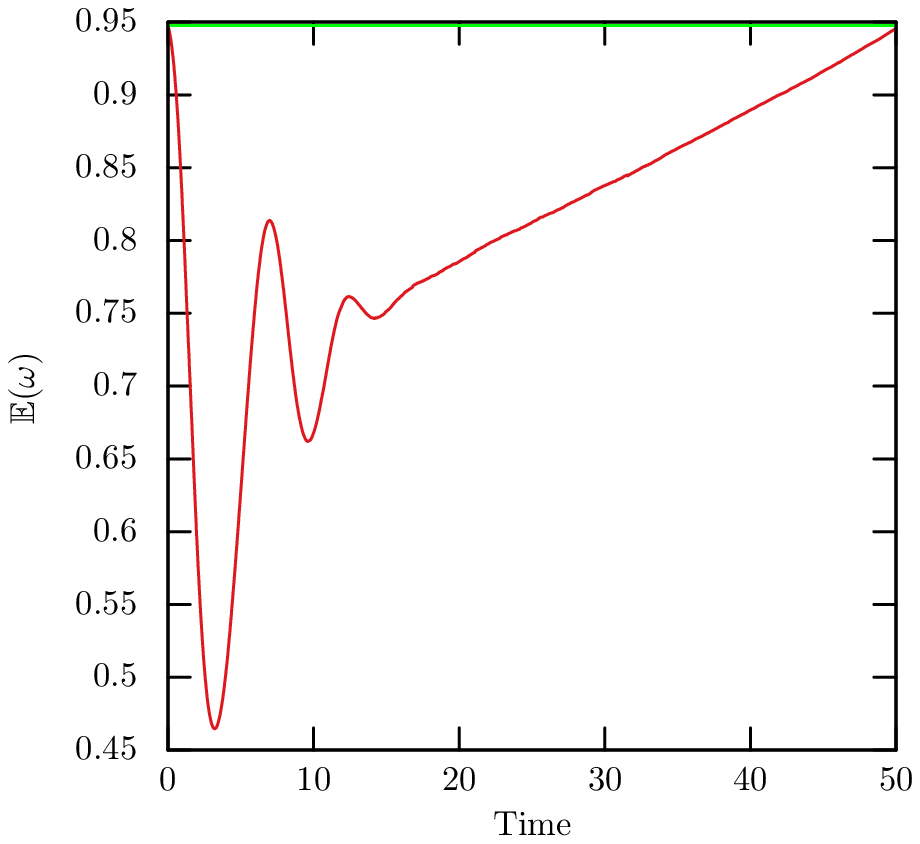}}} \\
			\end{tabular}
		}	
		\centering \caption{Second case: Expectation of the orbital elements.}
		\label{Egauss_detsto}
\end{figure}

On this example, and the one concerning the perturbation of the two-body problem (see \cite{cpp}), we can see that the use of It\^o's interpretation of white noises allows obtaining all the information contains in these objects. Indeed, considering white noises as basic functions of the time in the classical Gauss equations induces a lost of the information contained especially in the second derivatives. Whereas, we saw that it produces a non negligible effect on the dynamics and especially on the probabilistic mean behavior in the second case.
 
\section{A further extension with the Laplace-Runge-Lenz vector $\textbf{A}$}
We saw that the derivation of variations of orbital elements can be done with the angular momentum per unit of mass and the energy per unit of mass. Instead of the energy per unit of mass, we can use the Laplace-Runge-Lenz vector. Indeed, on the unperturbed orbit, this vector is defined by 
\begin{equation}
\label{runge-lenz}
\textbf{A}= \textbf{v} \times \textbf{H}  - \mu \frac{\textbf{r}}{r}
\end{equation}
and is conserved. The Laplace-Runge-Lenz gives the following relations $A = \mu e$, $E= \frac{A^2-\mu^2}{2H^2}$ and $a= \frac{\mu H^2}{\mu^2 - A^2}$ where $A$ is the norm of $\textbf{A}$. In consequence, the semi-major axis and the eccentricity are directly related to the Laplace-Runge-Lenz vector. It contains also the information of the pericenter location. Indeed, if the inclination is not zero we have
\begin{equation}
\tan \omega = \frac{H A_z}{H_x A_y - H_y A_x}
\end{equation}
and if the inclination is zero we have
\begin{equation}
\tan (\Omega + \omega) = \frac{A_y}{A_x}.
\end{equation}
In \cite{cpp}, this last relation were used, assuming that $\Omega=0$ , in order to derive the variation of the pericenter angle in the planar case. \\

We can also use the angular momentum vector per unit of mass and the Laplace-Runge-Lenz vector instead of the orbital elements $a$,$e$,$i$,$\Omega$ and $\omega$ as in \cite{roy-moran} for the deterministic case. Indeed, the equations governing the variations of these two vectors, hold for all kind of orbits. Thus, it is straightforward to derive the equations governing the variations of the orbital element. Even if the two vectors provide six components, they are not independent but related by the expression $\textbf{H}\cdot\textbf{A}=0$. In consequence, depending on which problem is studied, multiple choices are possible for the last element such as the true longitude which is the one chosen in \cite{roy-moran}. \\

We compute the variation of the Laplace-Runge-Lenz vector in order to have the set of perturbed equations $d\textbf{H}$ and $d\textbf{A}$ in the stochastic case. Using It\^o's formula, we obtain
\begin{equation}
d\textbf{A} = d\textbf{v} \times \textbf{H} + \textbf{v} \times d\textbf{H} + d\textbf{v} \times d\textbf{H} - \frac{\mu}{r} d\textbf{r} + \mu \frac{\textbf{r}\cdot \textbf{v}}{r^3} \textbf{r} dt.
\end{equation}
Then, using the expressions of $d\textbf{v}$, we obtain
\begin{equation}
d\textbf{A} = d\textbf{v}_P \times \textbf{H} + \textbf{v} \times d\textbf{H} + d\textbf{v} \times d\textbf{H}.
\end{equation}
Using the expression of $d\textbf{H}$ and $d\textbf{v}_P$, we obtain
\begin{equation*}
d\textbf{v} \times d\textbf{H} = \left[\mathsf{Tr}\left(\tilde{\textbf{v}}_P^\mathsf{T}\cdot\tilde{\textbf{v}}_P\right) \cdot \textbf{r}-\left(\tilde{\textbf{v}}_P\cdot \tilde{\textbf{v}}_P^\mathsf{T}\right)\cdot \textbf{r}\right]dt.
\end{equation*}
In order to write in the differential form the variation of the Laplace-Runge-Lenz vector, we define the operator $L: \textbf{u} \longmapsto L(\textbf{u})$ for any three dimensional vector, where $L(\textbf{u})$ is a three dimension square matrix with
\begin{equation*}
L(\textbf{u})=\left(\begin{array}{ccc}
0 & -u_3 & u_2\\
u_3  & 0 & -u_1 \\
-u_2 & u_1 & 0
\end{array}\right).
\end{equation*}
Then, for any another three dimensional vector $\textbf{v}$, we have $\textbf{u} \times \textbf{v} = L(\textbf{u})\cdot \textbf{v}$. Finally, we obtain
\begin{align*}
d\textbf{A} = &\left[\mathsf{Tr}\left(\tilde{\textbf{v}}_P^\mathsf{T}\cdot\tilde{\textbf{v}}_P\right) \cdot \textbf{r}-\left(\tilde{\textbf{v}}_P\cdot \tilde{\textbf{v}}_P^\mathsf{T}\right)\cdot \textbf{r} + \bar{\textbf{v}}_P \times \textbf{H} + \textbf{v}\times\left(\textbf{r}\times \bar{\textbf{v}}_P\right)\right]dt \\
& +\left[ \left(L(\textbf{v})\cdot L(\textbf{r})-L(\textbf{H})\right)\cdot\tilde{\textbf{v}}_P \right]\cdot d\textbf{B}.
\end{align*}
These last expression of Laplace-Runge-Lenz variation vector is also very convenient for numerical integration.

\section{Conclusion}

In this article, we have developed the stochastic perturbation equations of celestial mechanics which generalize the classical Gauss equations. This is done with the It\^o theory of stochastic differential equations and with basic considerations on the angular momentum and the energy per unit of mass. This approach allows predicting the impact of each components of the stochastic perturbing force on the dynamic. From a perturbing acceleration containing white noises, we showed the construction of the stochastic perturbation associated and we illustrated numerically the dynamic associated with the stochastic Gauss equations. Finally, we derived the variation of the Laplace-Runge-Lenz vector in order to obtain the minimum set of equations covering a large class of problem in celestial mechanics for further studies and applications.

\section{Acknowledgment}
I would like to thank the reviewers for their insightful comments on the paper which led me to an improvement of this work. I would also like to thank Jacky Cresson, Florent Deleflie and Lucie Maquet for their careful proofreading and discussions.

\appendix
\section{Proof of the stochastic Gauss equations}

In what follow, we always simplify computations in terms of orbital elements using the formulas from \eqref{eqr} to \eqref{eqsineps}. Moreover, we denote by $\tE=v\tR+rw\tT$ and $\tH=r\tT$ the stochastic part of the variation of the energy and the angular momentum (see Equation \eqref{dH} and \eqref{dE}). In the same way, we define for all the orbital elements and the angular momentum vector components, the quantities $\ta,\te,\ti,\tOmega,\tomega,\tH_x,\tH_y$ and $\tH_z$ to be the stochastic part in their variation.

\subsection{Semi-major axis $a$}
\label{dem_da}
We use the relation (\ref{eqEorb}) linking the energy $E$ and the semi-major axis $a$ in order to have 
\begin{align}
a=-\frac{\mu}{2E}.
\end{align}
Using It\^o's formula on the previous equation gives
\begin{align*}
da = \frac{\mu}{2 E^2}dE - \mu \frac{\tE\cdot\tE}{2 E^3}dt.
\end{align*}
Using the expression of the variation of the energy $E$ we obtain the result for $da$.

\subsection{Eccentricity $e$}
\label{dem_de}
Using It\^o's formula on Equation (\ref{eqe}), we obtain
\begin{align*}
de=\frac{2HE}{e\mu^2}dH+\frac{H^2}{e\mu^2}dE + \left(\frac{E}{e^3\mu^2}\tH\cdot\tH - \frac{H^4}{2e^3\mu^4}\tE\cdot\tE + \frac{2H(H^2E+\mu^2)}{e^3\mu^4}\tH\cdot\tE\right)dt.
\end{align*}
First, notice that
\begin{align*}
1+e\cos f-\frac{(1-e^2)}{1+e\cos f}=e\left(\cos f + \frac{e+\cos f}{1+e\cos f}\right)
\end{align*}
then
\begin{align*}
\frac{2HE}{e\mu^2}dH+\frac{H^2}{e\mu^2}dE = &\left[ \sqrt{\frac{a(1-e^2)}{\mu}}\left(\sin f \bar{R} + \left(\cos f + \frac{e+\cos f}{1+e\cos f}\right)\bar{T} \right) + \frac{a(1-e^2)}{2e\mu} \left(\tR^2+\tT^2\right)\right]dt \\
& +\sqrt{\frac{a(1-e^2)}{\mu}}\left(\sin f \tR + \left(\cos f + \frac{e+\cos f}{1+e\cos f}\right)\tT \right)\cdot d\textbf{B} .
\end{align*}
Second, using the expression of $dH$ and $dE$ we obtain 
\begin{align*}
\tH\cdot\tH &= \frac{a^2(1-e^2)^2}{(1+e\cos f)^2}\tT^2, \\
\tE\cdot\tE &= \frac{\mu e^2\sin^2 f}{a(1-e^2)}\tR^2 + \frac{\mu(1+e\cos f)^2}{a(1-e^2)}\tT^2+\frac{2e\mu\sin f (1+e\cos f)}{a(1-e^2)}\tR\cdot\tT,\\
\tH\cdot\tE &= \sqrt{\mu a(1-e^2)}\tT^2+\frac{e\sin f \sqrt{\mu a(1-e^2)}}{1+e\cos f}\tR\cdot\tT.
\end{align*}
Finally, after simplifications we obtain the result for $de$.

\subsection{Inclination $i$ and Ascending node $\Omega$}
\label{dem_didW}
In what follows, we assume that $i$ is not equal to zero. The variation of the inclination and the ascending node are related to the variation of the angular momentum vector $\textbf{H}$. We compute firstly the variation of the vector $\textbf{H}$. Using It\^o's formula, we obtain
\begin{align*}
d \textbf{H} = d\textbf{r} \times \textbf{v} + \textbf{r} \times d\textbf{v} + d\textbf{r} \times d\textbf{v}.
\end{align*}
Then, using the perturbed equations of motions (\ref{dv})-(\ref{dw}) we obtain
\begin{equation}
d\textbf{H} = \textbf{r} \times d\textbf{v}_P \label{dvecH}.
\end{equation}
Finally,
\begin{equation}
d\textbf{H} = -r(\bN dt+\tN \cdot d\textbf{B}) \textbf{e}_T + r(\bT dt +\tT \cdot d\textbf{B} ) \textbf{e}_N.
\end{equation}
The expression of $d \textbf{H}$ in the inertial frame is obtained as using three rotations (see Figure \ref{fig1})
\begin{align}
d \textbf{H} &=r \bigg[\bigg( \sin i \sin \Omega (\bT dt  + \tT \cdot d\textbf{B} ) \nonumber \\
&+ (\bN dt + \tN\cdot d\textbf{B} ) (\cos i \sin \Omega \cos (f+\omega)+\cos \Omega \sin (f+\omega))\bigg)\textbf{e}_x  \nonumber \\
&-\bigg( \sin i \cos \Omega (\bT dt  + \tT \cdot d\textbf{B} ) \nonumber \\
&+ (\bN dt + \tN\cdot d\textbf{B} ) (\cos i \cos \Omega \cos (f+\omega)-\sin \Omega \sin (f+\omega))\bigg)\textbf{e}_y  \nonumber \\
&+\bigg( \cos i (\bT dt  + \tT \cdot d\textbf{B} )- \sin i \cos (f+\omega) (\bN dt + \tN\cdot d\textbf{B} )\bigg)\textbf{e}_z \bigg] \nonumber 
\end{align}
with
\begin{align*}
\tH_x &=r \left( \sin i \sin \Omega \tT + \tN \left(\cos i \sin \Omega \cos (f+\omega)+\cos \Omega \sin (f+\omega)\right)\right),  \\
\tH_y &=r\left(-\sin i \cos \Omega \tT - \tN \left(\cos i \cos \Omega \cos (f+\omega)-\sin \Omega \sin (f+\omega)\right)\right),  \\
\tH_z &=r\left(\cos i\tT - \sin i \cos (f+\omega)\tN \right).
\end{align*}
Now we can compute the variation of the inclination $i$. Using It\^o's formula on Equation (\ref{eqcosi}), we obtain
\begin{align*}
-\sin idi-\frac{1}{2} \cos i (\ti\cdot\ti) dt= \frac{H_z}{H^2}dH-\frac{dH_z}{H}+ \left(\frac{\tH_z \cdot \tH}{H^2}-\frac{H_z(\tH\cdot\tH)}{H^3}\right)dt
\end{align*}
and so
\begin{align*}
di = -\frac{H^3 \cos i (\ti \cdot \ti)-2 H (\tH \cdot \tH_z )+2 H_z (\tH \cdot \tH)}{2 H^3 \sin{i}} dt + \frac{H_z}{H^2 \sin{i}}dH -\frac{1}{H \sin{i}} dH_z.
\end{align*}
Using the expression of $dH_z$ and $dH$, we finally obtain the result for $di$. Next, we compute the variation of the ascending node $\Omega$. Using It\^o's formula on Equation (\ref{eqtanOmega}), we obtain
\begin{align*}
\frac{d\Omega}{\cos^2 \Omega} + \frac{ (\tOmega \cdot \tOmega) \tan \Omega}{\cos^2 \Omega}dt = -\frac{dH_x}{H_y}+\frac{H_x}{H_y^2}dH_y + \left(\frac{\tH_x\cdot\tH_y}{H_y^2} -\frac{ H_x (\tH_y\cdot\tH_y)}{H_y^3}\right)dt
\end{align*}
and so
\begin{align*}
d\Omega &= \left(\frac{-H_x \cos ^2\Omega(\tH_y \cdot \tH_y)}{H_y^3}+ \frac{ \cos ^2\Omega(\tH_x \cdot \tH_y)}{H_y^2} -\tan \Omega(\tOmega \cdot \tOmega) \right) dt  \\
&-\frac{\cos ^2\Omega}{H_y}dH_x + \frac{H_x \cos ^2\Omega}{H_y^2} dH_y .
\end{align*}
Using the expression of $H_x$,$H_y$ and $dH_x$,$dH_y$, we can simplify the expression as
\begin{align*}
\frac{\cos^2\Omega}{H_y^2}(H_x dH_y - H_y dH_x) = \frac{r \sin(f+\omega)}{H\sin i}\bN dt + \frac{r \sin(f+\omega)}{H\sin i}\tN\cdot d\textbf{B}.
\end{align*}
After simplifications we obtain the result for $d\Omega$.

\subsection{Pericenter $\omega$}
\label{dem_dw}
In order to derive the variation of the pericenter location, we compute firstly the variation of the true anomaly $f$ and secondly the variation of the position angle $\theta$. Using It\^o's formula on Equation (\ref{eqr}), we obtain
\begin{align*}
\frac{2 H}{\mu r}dH +\left(\frac{\tH\cdot\tH}{\mu r}-\frac{H^2 v}{\mu r^2}\right)dt = \cos fde -e \sin fdf + \left(-\frac{1}{2} e \cos f (\tf\cdot\tf) -\sin f (\tE\cdot\tf)\right)dt
\end{align*}
and so
\begin{align*}
df= \left( \frac{H^2 v }{e \mu r^2 \sin{f}}-\frac{\tH^2 }{e \mu r\sin{f}}-\frac{\tE \cdot \tf}{e}-\frac{1}{2}\cot f (\tf\cdot\tf)\right) dt + \frac{\cot f}{e} de -\frac{2 H }{e \mu r \sin{f}} dH.
\end{align*}
Using the expression of the variation of the angular momentum $dH$, we obtain
\begin{align*}
df = \left( -\frac{2 H \bar{T} }{e \mu\sin{f}} -\frac{\tE \cdot \tf}{e}+\frac{H v w }{e \mu\sin{f}}-\frac{r \tT^2 }{e \mu\sin{f}}-\frac{1}{2} \cot f (\tf \cdot \tf)\right)dt -\frac{2 H}{e \mu \sin{f}}\tT \cdot d\textbf{B}  + \frac{\cot f}{e} de.
\end{align*}
Finally, using the expression of $de$ and after simplifications we obtain
\begin{align}
df &= \bigg[ \sqrt{\frac{a(1-e^2)}{\mu}}\frac{1}{e}\left(\cos f \bR - \sin f \left(\frac{2+e\cos f}{1+e\cos f} \right)\bT\right) + \frac{\sqrt{\mu}}{(a(1-e^2))^{3/2}}(1+e\cos f)^2  \nonumber \\
&+ \frac{a(1-e^2)}{\mu e^2}\bigg( -\frac{\sin 2f}{2} \tR^2 + \left(e+\cos f(2+e\cos f)^2 \right)\frac{\sin f}{(1+e\cos f)^2} \tT^2 \nonumber \\
&- \left(\frac{2+e\cos f}{1+e\cos f}\right)\cos 2f \tR \cdot \tT \bigg)\bigg]dt \nonumber \\
&+\sqrt{\frac{a(1-e^2)}{\mu}}\frac{1}{e}\left(\cos f \tR - \sin f \left(\frac{2+e\cos f}{1+e\cos f} \right)\tT\right) \cdot d\textbf{B} .
\end{align}
In order to compute the variation of the position angle, we use the z-component of the vector $d\textbf{r}$ and we use the It\^o's formula on the z-component of $d\textbf{r}$. We have
\begin{align*}
d(r \sin i \sin \theta)=(r w \sin i \cos \theta+v \sin i \sin \theta)dt
\end{align*}
which leads to
\begin{align*}
& r \cos i \sin \theta di +r \sin i \cos \theta d\theta  \\
&+\left(r \cos i \cos \theta(\ti\cdot\tthet)-\frac{1}{2} r \sin i \sin \theta(\ti\cdot\ti)-\frac{1}{2} r \sin i \sin \theta (\tthet\cdot\tthet)+v \sin i \sin \theta\right)dt  \\
&= (r w \sin i \cos \theta+v \sin i \sin \theta)dt.
\end{align*}
So we obtain
\begin{align*}
d\theta = \left(w-\cot i (\ti\cdot \tthet)+\frac{1}{2}\tan \theta (\ti \cdot \ti+\tthet \cdot \tthet) \right) dt -\cot i \tan \theta di.
\end{align*}
Using the expression of $di$ and after simplifications, we obtain
\begin{align}
d\theta &=\bigg[ -\sqrt{\frac{a(1-e^2)}{\mu}}\frac{\sin (f+\omega)\cot i}{1+e\cos f} \bN + \frac{\sqrt{\mu}}{(a(1-e^2))^{3/2}}(1+e\cos f)^2 \nonumber \\
&+\frac{a(1-e^2)}{2\mu (1+e\cos f)^2}\frac{\tan (f+\omega)}{\sin^2 i} \left( \cos^2 (f+\omega) \left(1 + \cos^2 i\right) +  \cos^2 i \right) \tN^2 \nonumber \\
&+\frac{a(1-e^2)}{\mu (1+e\cos f)^2} \cot i \sin (f+\omega) \tT \cdot \tN \bigg]dt \nonumber \\
&-\sqrt{\frac{a(1-e^2)}{\mu}}\frac{\sin (f+\omega)\cot i}{1+e\cos f} \tN \cdot d\textbf{B} .
\end{align}
Remarking that
\begin{equation}
d\theta=\frac{\sqrt{\mu}(1+e\cos f)^2}{(a(1-e^2))^{3/2}}dt +\frac{a(1-e^2)\sin(2(f+\omega))}{4\mu(1+e\cos f)^2} \tN^2 -\cos i d\Omega,
\end{equation}
we can deduce the variation of the pericenter location from the equation \eqref{eqf}.

\subsection{Mean anomaly $M$}
\label{dem_dM}
Using It\^o's formula on Equation (\ref{eqkepbis}), we obtain
\begin{align}
dM = &\left[\frac{(6 \cos f+e (5+\cos (2 f))) \sin f}{4 \sqrt{1-e^2} (1+e \cos f)^3} \left(\te \cdot \te\right) + \frac{e \left(1-e^2\right)^{3/2} \sin f}{(1+e \cos f)^3} \left(\tf\cdot\tf \right) \right. \nonumber \\
&\left. -\frac{\sqrt{1-e^2} \left(3 e+\left(2+e^2\right) \cos f\right)}{(1+e \cos f)^3} \left(\te \cdot \tf\right)	\right]dt \nonumber \\
& -\frac{\sqrt{1-e^2}\sin f (2+e \cos f)}{(1+e \cos f)^2}de-\frac{\left(1-e^2\right)^{3/2}}{(1+e \cos f)^2} df
\end{align}
with
\begin{align*}
\te \cdot \te = &\frac{\left(a \left(1-e^2\right) \sin ^2f\right) }{\mu }\tR^2+\frac{\left(a \left(1-e^2\right) (e+\cos f)^2\right)}{\mu  (1+e \cos f)^2}\tT^2 \\
&+\frac{\left(2 a \left(1-e^2\right) (e+\cos f) \sin f\right) }{\mu  (1+e \cos f)}\tR\cdot \tT \\
\end{align*}
\begin{align*}
\tf\cdot\tf =&\frac{a \left(1-e^2\right)\cos ^2f}{\mu e^2}\tR^2+\frac{a \left(1-e^2\right) (2+e \cos f)^2 \sin ^2f}{\mu e^2 (1+e \cos f)^2}\tT^2\\
&-\frac{a \left(1-e^2\right) (2+e \cos f) \sin (2 f)}{\mu e^2 (1+e \cos f)}\tR\cdot\tT \\
\end{align*}
\begin{align*}
\te \cdot \tf =& \frac{a \left(1-e^2\right) \cos f \sin f}{\mu  e}\tR^2-\frac{a \left(1-e^2\right)(e+\cos f) (2+e \cos f) \sin f}{2 \mu  e (1+e \cos f)^2}\tT^2 \\
&+\frac{a \left(1-e^2\right) \left(3 \cos (2 f)+2 e \cos ^3f-1\right)}{2 \mu  e (1+e \cos f)}\tR\cdot\tT.
\end{align*}
Remarking that
\begin{equation}
df=\frac{\sqrt{\mu}(1+e\cos f)^2}{(a(1-e^2))^{3/2}}dt +\frac{a(1-e^2)\sin(2(f+\omega))}{4\mu(1+e\cos f)^2} \tN^2 -(d\omega+\cos i d\Omega),
\end{equation}
we obtain after simplifications the result for $dM$.

\end{document}